\documentclass[12pt]{article}
\usepackage[colorlinks,linkcolor=Cobalt,citecolor=Cobalt,urlcolor=Cobalt,bookmarks,bookmarksnumbered]{hyperref}
\usepackage{mathpazo,charter}
\usepackage[scaled=0.95]{helvet}
\usepackage{amsmath,XXXF}
\usepackage{graphicx,color}
\usepackage{booktabs,multirow}

\newcommand{\Cl}{\mathop{\hbox{\rm Cl}}\nolimits}
\def\vdt{\partial_\tau}

\def\cM{\mathcal{M}}
\def\Gr#1{{\color{grey3}#1}}
\def\cb#1#2{\setlength\fboxsep{1pt}\colorbox{#1}{\color{#1}\fbox{\color{black}#2}}}
\def\cB#1{\hbox to0pt{\setlength\fboxsep{0pt}\hss\color{grey3}\fbox{\cb{white}{#1}}\hss}}
\def\bB#1{\hbox to0pt{\setlength\fboxsep{0pt}\hss\color{grey3}\fbox{\cb{black}{\color{white}#1}}\hss}}

 \SfTitles
 \allowdisplaybreaks
\begin{document}

\thispagestyle{empty}
 \noindent
 \today
  \vspace*{5mm}
 \begin{center}
{\LARGE\sf\bfseries
        On General Off-Shell Representations\protect\\*[4pt]
        of Worldline (1D) Supersymmetry}\\*[5mm]
{\large  Charles F. Doran,$^{\!\!a}$~
         Tristan H\"{u}bsch,$^{\!\!b}$~
         Kevin M.~Iga$^c$~
         and~
         Gregory D.~Landweber$^d$}\\[3mm]
{\small\it
  $^a$\,Dept. of Mathematics, University of Alberta, Edmonton, Alberta;
    \texttt{\slshape doran@math.ualberta.ca}\\*[0pt]
  $^b$\,Department of Physics \&\ Astronomy, Howard University, Washington, DC 20059;
    \texttt{\slshape thubsch@howard.edu}\\*[0pt]
  $^c$\,Natural Science Division, Pepperdine University, Malibu, CA 90263;
    \texttt{\slshape kiga@pepperdine.edu}\\*[0pt]
  $^d$\,Mathematics Program, Bard College, Annandale-on-Hudson, NY 12504-5000;
    \texttt{\slshape gregland@bard.edu}
 }\\[0mm]
  \vspace*{5mm}
{\sf\bfseries ABSTRACT}\\[2mm]
\parbox{158mm}{Every finite-dimensional unitary representation of the $N$-extended worldline supersymmetry without central charges may be obtained by a sequence of differential transformations from a direct sum of {\em\/minimal Adinkras\/}, simple supermultiplets that are identifiable with representations of the Clifford algebra. The data specifying this procedure is a sequence of subspaces of the direct sum of Adinkras, which then opens an avenue for classification of the continuum of so constructed off-shell supermultiplets.
 } 
\end{center}
\vspace{5mm}
\noindent
\parbox[t]{60mm}{PACS: {\tt11.30.Pb}, {\tt12.60.Jv}}\hfill
\parbox[t]{100mm}{\raggedleft\small\baselineskip=12pt\sl
             Once you begin rehearsal, then it's small building blocks.\\
             It's solving little problems one at a time.\\[-0pt]
            |~Michael Emerson}
\vspace{5mm}

\section{Introduction}
The study of off-shell supermultiplets in one dimension, \ie, finite-dimensional unitary off-shell representations of $N$-extended worldline supersymmetry,
 as  originally started in Refs.\cite{rGR0,rGR1,rGR2,rPT,rGLPR,rGLP},
 has been reinvigorated in the past decade or so\cite{rBIKL03,rBIKL04,rA,rBKMO,rKRT,r6-1,r6--1,rDI-SQM06,rDI-SQM06.2,rKT07,rILS,rDI-SQM07,rDI-SQM07.2,r6-3,r6-3.2,r6-1.2,rUMD09-1,rDI-SQM11,r6-3.1,rUMD09-2,rTHGK12,rGHHS-CLS,rGIKT12,rUMD12-3,rTHGK13}. In particular, an unprecedented abundance of new supermultiplets has been discovered\cite{r6-1.2,r6-3.1}, largely due to a graphical description of such supermultiplets using so-called Adinkras\cite{rA,r6-1}.
 Yet, in addition to all of these new supermultiplets, several off-shell supermultiplets have been analyzed recently\cite{rTHGK12,rGHHS-CLS,rGIKT12}, which cannot be depicted by Adinkras but represent various generalizations thereof; in fact, there are infinitely many such, more general supermultiplets\cite{rTHGK12}.

Herein, we prove that all off-shell worldline supermultiplets subject to physics-standard conditions can be analyzed in terms of Adinkras, as well as synthesized from them in a finite number of operations called {\em\/lowering\/}; see Theorem~\ref{T:main} in Section~\ref{s:Main}.

We then illustrate this procedure on two examples from the recent literature:
 ({\small\bf1})~the $N\,{=}\,3$ gauged worldline off-shell supermultiplet of Refs.\cite{rTHGK12,rTHGK13}, and
 ({\small\bf2})~the worldline (1d) shadow\cite{rGHHS-CLS} of the 4d ${\cal N}\,{=}\,1$ complex linear superfield\cite{r1001}.
Finally, we discuss a simple $N\,{=}\,5$ {\em\/continuous\/} family of worldline off-shell supermultiplets to see how, in principle, a general classification program of worldline ($1d$) off-shell supermultiplets could proceed.

\section{Supermultiplets}
We consider linear off-shell representations of the $N$-extended supersymmetry algebra in one dimension without central charge:
\begin{equation}
 \{Q_I,Q_J\}=2i\,\d_{IJ}\,\vdt
  \quad\text{and}\quad
 [\vdt,Q_I]=0\qquad
 \text{for}\quad I,J=1,{\cdots},N.
\label{e:SuSy}
\end{equation}
To be specific, linear representations of this algebra are spanned by a finite number of bosonic ($\f_a$) and fermionic ($\j_\a$) component fields (functions of time $\t$), upon which the operators $Q_I$ and $\vdt$ act linearly and so that the relations\eq{e:SuSy} are always satisfied.
 If the relations\eq{e:SuSy} are satisfied identically on a supermultiplet without requiring any of the component fields to satisfy any time-differential equation, the representation is said to be off-shell.

Any particular choice of component fields $(\f_a|\j_\a)$ used to represent a supermultiplet should be considered a basis of the supermultiplet, and is subject to redefinitions.
 Local linear changes of variables $\F_a=\F_a(\f,\vdt\f,{\cdots})$ and $\J_\a=\J_\a(\j,\vdt\j,{\cdots})$ produce equivalent representations of the same supermultiplet.
 Since the operators $Q_I$ are themselves fermionic, their application on any one component field must produce a linear combination of the fields of the opposite statistics and their $\t$-derivatives.
 It follows that a supermultiplet $\cM$ is a {\em\/closed $Q$-orbit\/}, in that the result of the application of any sequence of $Q_I$-operators on any component field of $\cM$ must result in a linear combination of the component fields in $\cM$ and their $\t$-derivatives.

We next define a few additional useful properties of off-shell supermultiplets, assuming only the supersymmetry algebra\eq{e:SuSy}. In particular, we make no assumption about any Lagrangian or intended dynamics for the supermultiplets. 

\subsection{Engineerable Supermultiplets}
As done in\eq{e:SuSy}, we adopt the ``natural'' system of units, fixing $c$ and $\hbar$ as two basic units, which we then never write explicitly\cite{rBjDr2,rPS-QFT}. The physical units of each quantity then reduce to a power of a single unit, which we choose to be mass.
 For a quantity $X$ with units $(\text{mass})^{\sss[X]}$, the exponent $[X]$ is called the {\em\/engineering dimension\/} of the quantity $X$.
 For an operator $\mathscr{O}$, we say that $[\mathscr{O}]=y$ if $[\mathscr{O}(X)]=y+[X]$, for all $X$ upon which the operator $\mathscr{O}$ is defined to act.

In particular, since time has engineering dimension $[\t]\,{=}\,{-}1$, $[\vdt]=1$. If the $Q_I$ in\eq{e:SuSy} are to have unambiguous engineering dimensions\ft{Abstract Hermitian operators may be assigned the physical units of the physical observable to which they correspond. This rule does not apply to fermionic operators, as they cannot correspond to any physical observable.}, the supersymmetry algebra\eq{e:SuSy} implies that $[Q_I]=\frac12$ for all $I=1,{\cdots},N$. 
 In the $Q$-transformation rules, the action of each $Q_I$ is specified on each component field. If $[F]=f$ for a component field $F$, $[Q_I(F)]=f{+}\frac12$, and $Q_I(F)$ is a linear combination of component fields of engineering dimension $f{+}\frac12$ as well as $\t$-derivatives of component fields of engineering dimension $f{-}\frac12$. Formalizing this, we have:
\begin{defn}
A supermultiplet is {\em\/engineerable\/} if it is possible to consistently assign engineering dimensions to all the component fields, such that $[Q_I(F)]=f{+}\frac12$ for every component field $F$ with $[F]=f$.
\end{defn}
All supermultiplets commonly used in the literature are engineerable. Nevertheless, this is a logical assumption that deserves to be spelled out since it will be used in the mathematical proof of our main theorem.
In turn, non-engineerable supermultiplets are logically possible if there is a fixed mass parameter. For instance, in Planck units (including Newton's gravitational constant $G_{\sss N}$ as an unwritten unit), non-engineerable supermultiplets are possible since the action of $Q_I$ on certain component fields may involve (unwritten) powers of the Planck mass.  As an aside, a non-engineerable ``Escheric'' supermultiplet was discussed in Ref.\cite{rA}.

\begin{prop}\label{P:B-F=1/2}
If an engineerable supermultiplet $\cM$ contains a boson $\f$ and a fermion $\j$ for which $[\j]-[\f]\neq\frc12\pmod\ZZ$, $\cal M$ must decompose as a direct sum of at least two supermultiplets, $\cM\supseteq\cM_1\oplus\cM_2$, such that $\f\in\cM_1$ and $\j\in\cM_2$.
\end{prop}
\begin{proof}
Let $[\f]=m$ be the engineering dimension of $\f\in\cM$. Then, $Q_I(\f)$ has engineering dimension $m{+}\frc12$ and is a linear combination of some of the fermions and their $\t$-derivatives. Since $[\vdt]=1$, each fermion occurring in $Q_I(\f)$ has the engineering dimension $m{+}\frc12\pmod\ZZ$.
 In turn, applying $Q_J$ on each of these fermions must produce a linear combination of bosons and their $\t$-derivatives, where each of these resulting bosons has the engineering dimension $m\pmod\ZZ$.
 Iterating this argument eventually maps out a sub-supermultiplet $\cM_1\subset\cal M$, wherein all bosons and fermions have the engineering dimensions $m\pmod\ZZ$ and $m{+}\frc12\pmod\ZZ$, respectively. By construction, $[\text{fermion}]-[\text{boson}]=\frc12\pmod\ZZ$ throughout $\cM_1$.

Since $[\j]-[\f]\neq\frc12\pmod\ZZ$, $\j\not\in\cM_1$ and $\j$ must belong to a separate sub-supermultiplet of $\cM$. Repeating the above construction starting with $\j$ maps out this other $\cM_2\subset\cM$.
\end{proof}
In most physics applications, the relation\ft{When not using the ``$\hbar=1=c$'' units, this relationship is $[\text{fermion}]-[\text{boson}]=[Q]+n[\hbar\vdt]$, with $n\in\ZZ$.} $[\j_\a]-[\f_a]+\frc12=0\pmod\ZZ$ is enforced by dynamical considerations even for (component) fields that are completely unrelated, by supersymmetry or otherwise. However, no bosonic-fermionic pair of fields satisfying $[\j]-[\f]\neq\frc12\pmod\ZZ$ can belong to an {\em\/indecomposable\/} supermultiplet. 

\paragraph{Nomenclature:} A representation (and so also a supermultiplet) is said to be {\em\/indecomposable\/} if it cannot be decomposed into a direct sum of two or more sub-representations. A representation is called {\em\/irreducible\/} if it contains no proper sub-representation. The latter condition is stronger, in that indecomposable representations need not be irreducible, whereas all irreducible representations are necessarily indecomposable.

For example, the well-known real vector superfield $\IV=\IV^\dag$ furnishes a representation of simple supersymmetry in 4-dimensional spacetime.
 Given a chiral superfield $\bm\L$ and its conjugate $\bm\L^\dag$, one can gauge away the ``lower half'' sub-representation, {\em\/reducing\/} $\IV$ to the gauge quotient $\{\IV\simeq\IV{+}\Imm(\bm\L)\}$ known as the ``vector superfield in the Wess-Zumino gauge'' and represented by the ``upper half'' component fields\cite{r1001,rBK,rWS-Fields}. This {\em\/reduces\/} the vector supermultiplet to the half-as-large gauge quotient although $\IV$ does not decompose into a direct sum of two complementary superfields: only the ``lower half'' of $\IV$ is a proper sub-superfield, identified with $\Imm(\bm\L)$ in the Wess-Zumino gauge. Thus, $\IV$ is reducible but indecomposable.

\subsection{Adinkraic Supermultiplets}
All off-shell worldline supermultiplets discussed in Refs.\cite{rA} and then formalized rigorously in Refs.\cite{r6-1,r6--1} admit a basis of component fields such that the application of any one supercharge $Q_I$ on any one component field always produces precisely one other of the component fields or its $\t$-derivative. Such supermultiplets were called {\em\/adinkraic\/} and can be depicted by {\em\/Adinkras\/}, specific graphs that faithfully encode the precise supersymmetry transformations within each such supermultiplet. By extension, a supermultiplet for which an adinkraic basis of component fields can be obtained by means of local component field redefinitions is called {\em\/adinkrizable\/}.

As Ref.\cite{r6-1} proves, in adinkraic supermultiplets it is always possible to rescale the (real) component fields so that (with no summations):
\begin{equation}
  Q_I(\f_a)  = \pm (\vdt^{\,e_{a,\a}}\j_\a),\quad
  Q_I(\j_\a) = \pm i(\vdt^{1-e_{a,\a}}\f_a),\quad
  e_{a,\a}\Defl\frc12{+}[\f_a]{-}[\j_\a].
 \label{e:A}
\end{equation}
This system of $Q$-transformations, complete with the specific sign-choices, is completely encoded by the graphical elements of an Adinkra\cite{r6-1}, which made the classification results of Refs.\cite{r6-3,r6-3.2,r6-3.1} possible, as well as the 1--1 translation into the familiar superfield framework\cite{r6-1.2}.
\ping

There of course exist non-adinkraic off-shell supermultiplets, which do not admit a basis of component fields wherein the simple pattern\eq{e:A} holds.
Nevertheless, recent study shows that all off-shell supermultiplets from the familiar literature on simple supersymmetry in 4-dimensional spacetime\cite{rUMD09-1,rUMD09-2,rGHHS-CLS,rUMD12-3} are either themselves adinkraic or can be described in terms of Adinkras. This is also true of at least several off-shell supermultiplets of ${\cal N}\,{=}\,2$-extended supersymmetry in 4-dimensional spacetime\cite{r6-4}, and also of the constructions of off-shell worldline supermultiplets as given in Refs.\cite{rGIKT12} and\cite{rTHGK12}, the latter one of which provides an infinite sequence of new, ever larger supermultiplets.

In all these examples of non-Adinkraic off-shell supermultiplets, the application of the supercharges $Q_I$ on individual component fields does not always result in the monomials\eq{e:A}, but requires a binomial or larger linear combination of component fields and their $\t$-derivatives. Also, all these non-adinkraic supermultiplets may be related to adinkraic ones by means of non-local field redefinitions. Our subsequent results will not only prove that this is in fact true of all off-shell worldline supermultiplets, but will also provide a constructive algorithm to this end.

\subsection{Valise Supermultiplets}
If the component fields in an engineerable supermultiplet have only two distinct engineering dimensions (differing by $\frc12$), one for bosons and another for fermions, we call this a {\em valise supermultiplet}. In Ref.~\cite{r6-3}, Adinkras with this property were called Isoscalar and Isospinor, for when bosons or fermions have a lower engineering dimension, respectively; the boson/fermion-indiscriminate moniker ``valise'' was adopted in Ref.~\cite{r6-3.2}, so this is a generalization of that terminology. This was also called the ``Clifford algebraic supermultiplet'' in Ref.\cite{rGLP} and plays prominent role in the ``root superfield'' formalism\cite{rGLP,rA,rBKMO}, which is then also generalized by the subsequent results.

We are now in position to state and prove the following important result.
\begin{thrm}[valises]\label{T:valise}
Every valise supermultiplet is adinkraic, and decomposes as a direct sum of minimal valise supermultiplets, each identifiable with a minimal valise Adinkra.
\end{thrm}

\begin{proof}
Without loss of generality, suppose the bosonic component fields in the valise supermultiplet have the lower engineering dimension.  Then if $\f$ is any boson, $Q_I(\f)$ is a linear combination of the fermions, while if $\j$ is any fermion, $Q_I(\j)=\vd_t(\ell)$ where $\ell$ is a linear combination of the bosons.

If we choose a basis $\vf_1,\ldots,\vf_d$ for the bosons and $\c_1,\ldots,\c_d$ for the fermions, these $Q$-transformation rules can be written as
\begin{equation}
  Q_I(\vf_a)  = \sum_\a [L_I]_a{}^\a\, \c_\a,\qquad\text{and}\qquad
  Q_I(\c_\a) = i\sum_a [R_I]_\a{}^a\, \dt\vf_a,
 \label{e:valise}
\end{equation}
where $[L_I]$ and $[R_I]$ are matrices. 
 Following the analysis in Ref.~\cite{r6-3.1}, we define the fermion counting operator $(-1)^F$ so that
\begin{equation}
   (-1)^F \f_a=\f_a
    \qquad\text{and}\qquad
   (-1)^F \j_\a=-\j_\a,
\end{equation}
and define the $2d\times2d$ matrices
\begin{equation}
 (-1)^F=
 \G_0\Defl\left[\begin{array}{c|c}
            \Ione & ~~\,\bm0\\ \hline
            \bm0  & {-}\Ione
          \end{array} \right]
 \qquad\text{and}\qquad
 \G_I\Defl\left[\begin{array}{c|c}
            \bm0 & L_I\\ \hline
            R_I & \bm0
          \end{array} \right]
 \quad\text{for }I=1,{\cdots}\,,N.
\end{equation}
These matrices, $\{\G_0,\G_1,\ldots,\G_N\}$, satisfy the relations for a Clifford algebra $\Cl(0,N{+}1)$, and so in this way, form a real matrix representation $\bm{M}$ of $\Cl(0,N{+}1)$. The relations\eq{e:valise} identify this matrix representation of $\Cl(0,N{+}1)$ with the supermultiplet $(\vf_1,{\cdots},\vf_d|\c_1,{\cdots},\c_d)$.

It is a standard result\cite{rLM} that all real representations of $\Cl(0,N{+}1)$ decompose into a direct sum of irreducible representations:  For $N\,{=}\,0\!\pmod4$, there exist two equal-sized but distinct isomorphism classes of irreducible representations\ft{ It is the existence of these two distinct isomorphism classes that makes supermultiplet twisting\cite{rTwSJG0,rGHR} nontrivial. As representations of $\Cl(0,N{+}1)$ also correspond to the somewhat more familiar spinor representations of $\text{Spin}(N)$, the Reader may recognize this fact by recalling that for $k,n\in\ZZ$, $\Spin(n,8k{+}n)$ have Majorana-Weyl minimal spinors, while $\Spin(n,4k{+}n)$ for odd $k$ have complex conjugate pairs of minimal spinors.}; otherwise, there exists only one.

It was shown in Refs.~\cite{r6-3.2,r6-3.1} that the irreducible representation(s) of $\Cl(0,N{+}1)$ are adinkrizable. That is, each of them corresponds 1--1 to a valise supermultiplet akin to\eq{e:valise}, but where the linear combinations in the result of applying the $Q_I$-operators reduce to monomials\ft{The $[L_I]$- and $[R_I]$-matrices then have a single nonzero entry in every row and column, and are called {\em\/monomial\/}.}. This valise supermultiplet then may be depicted by an Adinkra, by assigning a node to each component field and an $I$-colored edge for every instance of a relation $Q_I(\vf_a)=\pm i\c_\a$, drawing the edge solid for the positive sign and dashed for the negative sign. Furthermore, these irreducible representations of $\Cl(0,N{+}1)$ are clearly minimal, and so must correspond to minimal valise Adinkras and supermultiplets for any given $N$.

Now return to our real representation $\bm{M}$. It is a standard result\cite{r6-3.1} that, as a representation of $\Cl(0,N{+}1)$, $\bm{M}$ decomposes into a direct sum of irreducible representations.  For each direct summand, choose the basis given so that its corresponding valise supermultiplet is described using an Adinkra as per Ref.\cite{r6-3,r6-3.1}.  The result is a basis for $\bm{M}$ in which the valise supermultiplet\eq{e:valise} is adinkraic, and in fact decomposes as the corresponding direct sum of minimal Adinkras (for the given $N$).
\end{proof}
\noindent Without loss of generality, we identify an Adinkra with the supermultiplet it depicts.

It also follows that there exist two distinct isomorphism classes of minimal valise supermultiplets for $N\,{=}\,0\!\pmod4$, one being referred to as the twisted version of the other\cite{rTwSJG0,rGHR}. In turn, there is only one isomorphism class for $N\,{\neq}\,0\!\pmod4$.
 Refs.\cite{r6-3,r6-3.1} prove that
\begin{equation}
  \min\big(\dim(\text{Adinkra})\big) = (2^{N-\vk(N)-1}|2^{N-\vk(N)-1}),
 \label{e:mindim}
\end{equation}
where
\begin{equation}
  \vk(N)
  =\left\{\begin{array}{l@{~~\text{for}~}l}
                 0 & N=0,1,2,3;\\
                 1 & N=4,5;\\
                 2 & N=6;\\
                 3 & N=7;\\
                 4+\vk(N{-}8) & N\geqslant8.\\
          \end{array}\right.
\end{equation}
That is, minimal supermultiplets then have $2^{N-\vk(N)-1}$ real bosonic component fields, and as many fermionic ones.

\subsection{Raising and Lowering}
Refs.\cite{rGR1,rGR2} and then\cite{rGLP,rA,rBKMO} started exploring the systematic use of an operation variously called ``automorphic duality,'' ``1D duality'' and ``auxiliary/physical duality,'' of which a refinement (to individual component fields) was named ``node raising and lowering'' owing to its manifest depiction in terms of Adinkras\cite{r6-1}.
 These operations easily generalize to all engineerable supermultiplets:
\begin{defn}[raising/lowering]\label{D:R/L}
Let $\cM$ be an engineerable off-shell supermultiplet, and $\ell=\sum_Ac_AF_A$ a real linear combination of component fields of $\cM$, all with the same engineering dimension.
\begin{itemize}\addtolength{\leftskip}{1pc}\itemsep=-3pt\vspace{-4mm}
 \item[{\rm a.}] If $Q_I(\ell)$ involves no derivative of any component field of $\cM$ for any $I$, 
  replacing any one of $F_A\in\cM$ with $L=\skew3\dt{\ell}$ and assigning
  $Q_I(L)=\vd_t(Q_I(\ell))$ produces a new engineerable off-shell supermultiplet
  $\cM^\sharp_\ell$.
  This is called ``{\em\/raising $\bm\ell$\/}'' and $[L]=[\ell]{+}1$.
 \item[{\rm b.}] If $Q_I(\ell)=\vdt(f_I)$ is a total $\t$-derivative for each $I$,
  replacing any one of $F_A\in\cM$ with $L=\int\!\rd\t\,\ell$ and assigning
  $Q_I(\ell)=f_I$ produces a new engineerable off-shell supermultiplet $\cM^\flat_\ell$.
  This is called ``{\em\/lowering $\bm\ell$\/}'' and $[L]=[\ell]{-}1$.
\end{itemize}
\end{defn}
\paragraph{Note:} For the subsequent theorem, we will only need the special case of the raising operation when $\ell$ is in fact a single component field. There is no reason, however, not to provide the general definition.

If $\cM$ can be depicted by an Adinkra and $\ell$ is a single component field (represented by a node), the operations of raising/lowering $\ell$ then reduce to ``node raising/lowering''\cite{r6-1}, \ie, ``auxiliary/physical duality'' of Ref.\cite{rGLP}. When performed {\em\/simultaneously\/} on each one of the $\binom{N}k$ component fields equal to $Q_{I_1}{\cdots}Q_{I_k}(\f)$ for a fixed $\f$ and $k$, and replacing each of these component fields separately with another one of one unit higher or lower engineering dimension, these operations reproduce the ``automorphic duality'' of Ref.\cite{rGR1,rGR2,rGLP,rA}. A matrix realization of the node-raising operation was also introduced in Ref.\cite{rPT}, and was subsequently called the ``dressing transformation''\cite{rKRT}.

In the general cases covered by the above definitions, the linear combinations of fields being raised or lowered can extend over any subset of component fields of the same engineering dimension in the original off-shell supermultiplet $\cM$. The coefficients in the linear combination are here assumed to be real\ft{The definitions are straightforward to adapt to working over complex or hypercomplex supermultiplets.}, but are otherwise arbitrary and can be varied continuously.

\section{The Main Theorem}
 \label{s:Main}
The foregoing discussion was framed to state:
\begin{thrm}[adinkraic analysis/synthesis]\label{T:main}
Every engineerable off-shell worldline supermultiplet $\cM$ with finitely many component fields is equivalent by local field redefinitions to a supermultiplet obtained from a direct sum of minimal valise Adinkras, by iteratively lowering linear combinations of nodes.
\end{thrm}

\begin{proof}
Proposition~\ref{P:B-F=1/2} decomposes $\cM$ into parts in each of which $[\text{fermion}]-[\text{boson}]=\frc12\pmod\ZZ$; we work with each of these parts in turn. For simplicity, each such part on which we focus iteratively will continue to be denoted by $\cM$.

Consider the $2d$ component fields $\cM$ and their engineering dimensions. Let $m$ be the minimum, and $M$ the maximum of these engineering dimensions in $\cM$. If $M=m{+}\frc12$, $\cM$ is already a valise supermultiplet ($\cM=\cM^v$), skip to {\bsf Part~2}; if $M>m{+}\frc12$, proceed.

\paragraph{Part~1:}
Choose any one component field $f_1$ with engineering dimension $m$ (while $m<M{-}\frc12$) and raise it.  This results in a new engineerable supermultiplet $\cM^\sharp_1$.  The new raised field now has engineering dimension $m+1$, and if $M-m\ge 1$, then $\cM^\sharp_1$ will still have maximum engineering dimension $M$. Repeating this process reduces the number of component fields with engineering dimension $m$, until there are none.  Then we have a new supermultiplet with the minimum engineering dimension $m+\frac12$. Keep repeating this process, until the minimum engineering dimension increases to $M-\frac12$.

This transforms the original supermultiplet $\cM$ into an associated valise supermultiplet $\cM^v$, with a finite sequence of component fields $(\f_1,{\cdots}\,\f_d|\j_1,{\cdots}\,\j_d)$, where $[\f_a]=m=M{-}\frc12$ and $[\j_\a]=M$, or the other way around.

\paragraph{Part~2:}
By Theorem~\ref{T:valise}, the valise supermultiplet $\cM^v$ admits a basis $(\vf_1,{\cdots}\,\vf_d|\c_1,{\cdots},\c_d)$ where each of the component fields $(\vf_a|\c_\a)$ is a linear combination of the component fields $(\f_a|\j_\a)$, which decomposes $\cM^v\simeq\Tw\cM^v=\oplus_i\,\Tw\cM^v_i$ as a direct sum of minimal valise supermultiplets. Each minimal valise supermultiplet $\Tw\cM^v_i$ may be identified with a minimal valise Adinkra, and each of the basis elements $(\vf_a|\c_\a)$ with a node.

\paragraph{Part~3:}
Inverting the linear combinations from Part~2, the fields $\f_a,\j_\a\in\cM^v$ can now be written as linear combinations of the nodes $\vf_a,\c_\a\in\oplus_i\,\Tw\cM^v_i$.  Reversing the procedure of Part~1, we iteratively lower $\f_1$ (as a linear combination of the nodes $\vf_a$), then $\f_2$ and so on, until each of these the nodes is lowered to its original engineering dimension.

The result is the original supermultiplet $\cM$, reconstructed as a systematically iterated lowering of the direct sum of minimal valise Adinkras, $\oplus_i\,\Tw\cM^v_i$.
\end{proof}
\paragraph{Note:} It may well happen that the supermultiplet can decompose into a direct sum of minimal supermultiplets before Part~1 of the procedure in the proof is completed; see Section~\ref{s:GKY}.

This general result, true for all $N$-extended supersymmetry algebras without central extension\eq{e:SuSy}, covers the $N\,{=}\,4$ partial results obtained to date for a handful of supermultiplets obtained by dimensional reduction from simple supersymmetry in 4-dimensional spacetime: see Table~1 of Ref.\cite{rUMD09-2} and Tables~10--12 of Ref.\cite{rUMD12-3}.

\section{Non-Adinkraic Supermultiplets}
While the most familiar and oft-used supermultiplets of simple supersymmetry in 4-dimensional spacetime turn out to be adinkraic and are easily depicted using Adinkras\cite{rUMD09-1,rUMD09-2}, there do exist examples where this is not true. We now turn to such non-adinkraic examples, to demonstrate the effectiveness of Theorem~\ref{T:main}. The Adinkras or Adinkra-like graphs depicting the supersymmetry transformations will illustrate the procedure.

\subsection{A Gauge-Quotient Example}
 \label{s:GKY}
Ref.~\cite{rTHGK12} constructs an off-shell gauge-quotient supermultiplet of $N\,{=}\,3$ worldline supersymmetry, $\IY_I/(iD_I\IX)$, and shows it to be non-Adinkraic, in that a minimum of six transformation rules involve binomial combinations of fields; see\eq{e:GKY}, reproduced here in the simpler notation of Ref.~\cite{rTHGK13}:
\begin{equation}
  \begin{array}{@{} c|c@{~~}c@{~~}c @{}}
 & \C1{Q_1} & \C2{Q_2} & \C3{Q_3} \\ 
    \hline\rule{0pt}{2.1ex}
\f_1
 & \j_1 & \j_2 & \j_3 \\[0pt]
\f_2
 & \j_3 & -\j_4 & -\j_1 \\[0pt]
\f_3
 & \j_4\Gr{{-}\j_7} & \j_3\Gr{{-}\j_5}
   & -\j_2\Gr{{+}\j_6} \\[-2pt]
\f_4
 & \j_5 & -\j_7 & -\j_8 \\[0pt]
\f_5
 & -\j_6 & \j_8 & -\j_7 \\*[3pt]
    \cline{2-4}\rule{0pt}{2.5ex}
F_1
 & \dt\j_2 & -\dt\j_1 & \dt\j_4 \\[0pt]
F_2
 & \dt\j_8 & \dt\j_6 & \dt\j_5 \\[0pt]
F_3
 & \dt\j_7 & \dt\j_5 & -\dt\j_6 \\
    \hline
  \end{array}
 \qquad
  \begin{array}{@{} c|c@{~~}c@{~~}c @{}}
 & \C1{Q_1} & \C2{Q_2} & \C3{Q_3} \\ 
    \hline\rule{0pt}{2.9ex}
\j_1
 & i\dt\f_1 & -iF_1 & -i\dt\f_2 \\[0pt]
\j_2
 & iF_1 & i\dt\f_1 & -i\dt\f_3\Gr{{-}iF_3} \\[0pt]
\j_3
 & i\dt\f_2 & i\dt\f_3\Gr{{+}iF_3} & i\dt\f_1 \\[0pt]
\j_4
 & i\dt\f_3\Gr{{+}iF_3} & -i\dt\f_2 & iF_1 \\[0pt]
\j_5
 & i\dt\f_4 & iF_3 & iF_2 \\[0pt]
\j_6
 & -i\dt\f_5 & iF_2 & -iF_3 \\[0pt]
\j_7
 & iF_3 & -i\dt\f_4 & -i\dt\f_5 \\[0pt]
\j_8
 & iF_2 & i\dt\f_5 & -i\dt\f_4 \\[0pt]
    \hline
  \end{array}
  \label{e:GKY}
\end{equation}
Nevertheless, the transformation rules\eq{e:GKY} {\em\/can\/} be represented graphically, but doing so requires that several edges (depicted by tapering lines) to indicate a ``one way'' $Q_I$-transformation; see Figure~\ref{f:GKY}.
\begin{figure}[ht]
\centering
 \begin{picture}(140,45)
   \put(0,0){\includegraphics[width=140mm]{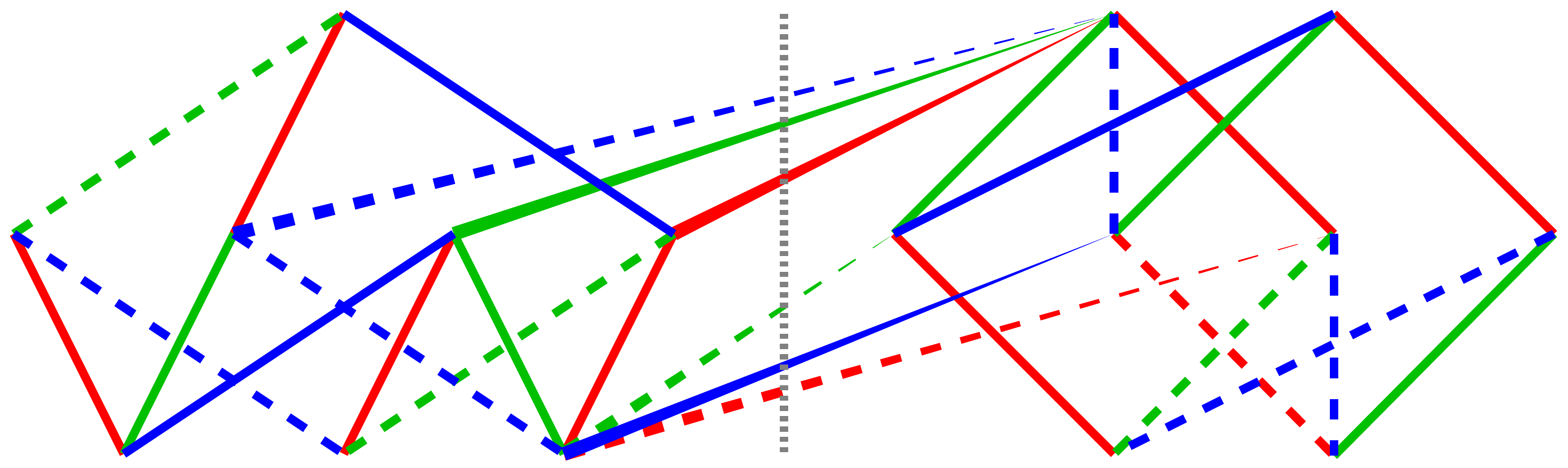}}
    \put(12,1){\cB{$\f_1$}}
    \put(31,1){\cB{$\f_2$}}
    \put(50.5,1){\cB{$\f_3$}}
    \put(99.5,1){\cB{$\f_4$}}
    \put(119,1){\cB{$\f_5$}}
    \put(2,20){\bB{$\j_1$}}
    \put(21,20){\bB{$\j_2$}}
    \put(40,20){\bB{$\j_3$}}
    \put(59,20){\bB{$\j_4$}}
    \put(80,20){\bB{$\j_5$}}
    \put(99,20){\bB{$\j_6$}}
    \put(118,20){\bB{$\j_7$}}
    \put(137,20){\bB{$\j_8$}}
    \put(31,40){\cB{$F_1$}}
    \put(99.5,40){\cB{$F_3$}}
    \put(119,40){\cB{$F_2$}}
 \end{picture}
\caption{A graphical depiction of the gauge-quotient supermultiplet of Ref.\protect\cite{rTHGK12}.}
 \label{f:GKY}
\end{figure}
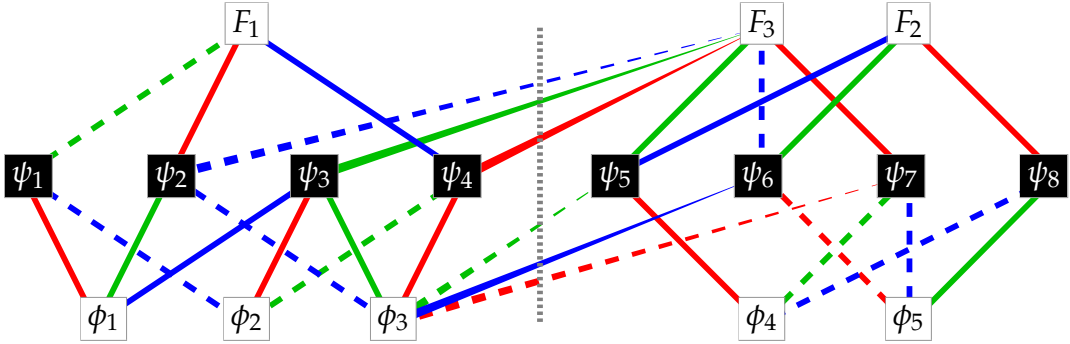
These correspond to the entries in the table\eq{e:GKY} that are set in a lighter ink. For example, $\C1{Q_1}(\f_3)$ includes both $\j_4$ and $\j_7$, but only the transformation of $\C1{Q_1}(\j_4)$ contains $\dt\f_3$, the transformation $\C1{Q_1}(\j_7)$ doesn't. The two-way $\C1{Q_1}$-transformation $\f_3\iff\j_4$ is then depicted by a standard edge\ft{Edges drawn in the $I^\text{th}$ color signify $Q_I$-transformation: solid for a positive sign and dashed for a negative sign.}, while the one-way $\C1{Q_1}$-transformation $\f_3\to{-}\j_7$ is depicted by a tapering edge, which is also dashed, indicating the negative sign.

Following the constructive proof of Theorem~\ref{T:main}, in this supermultiplet we:
 ({\small\bf1})~raise the node corresponding to $\f_3\mapsto\dt\f_3$ (this loses the constant term in $\f_3$), and
 ({\small\bf2})~perform the linear combination change of variables $\dt\f_3\mapsto Z\Defl\dt\f_3{+}F_3$, the net effect of which is the field substitution
\begin{equation}
   \f_3~~\mapsto~~Z\Defl\dt\f_3{+}F_3.
 \label{e:f3r}
\end{equation}
This produces an associated supermultiplet:
\begin{equation}
  \begin{array}{@{} c|c@{~~}c@{~~}c @{}}
 & \C1{Q_1} & \C2{Q_2} & \C3{Q_3} \\ 
    \toprule
\f_1
 & \j_1 & \j_2 & \j_3 \\[-1pt]
\f_2
 & \j_3 & -\j_4 & -\j_1 \\[-1pt]
\f_4
 & \j_5 & -\j_7 & -\j_8 \\[-1pt]
\f_5
 & -\j_6 & \j_8 & -\j_7 \\*[3pt]
    \cline{2-4}\rule{0pt}{2.5ex}
F_1
 & \dt\j_2 & -\dt\j_1 & \dt\j_4 \\[1pt]
F_2
 & \dt\j_8 & \dt\j_6 & \dt\j_5 \\[1pt]
F_3
 & \dt\j_7 & \dt\j_5 & -\dt\j_6 \\[1pt]
Z
 & \dt\j_4 & \dt\j_3  & -\dt\j_2 \\[2pt]
    \bottomrule
  \end{array}
 \qquad
  \begin{array}{@{} c|c@{~~}c@{~~}c @{}}
 & \C1{Q_1} & \C2{Q_2} & \C3{Q_3} \\ 
    \toprule
\j_1
 & i\dt\f_1 & -iF_1 & -i\dt\f_2 \\[0pt]
\j_2
 & iF_1 & i\dt\f_1 & -iZ \\[0pt]
\j_3
 & i\dt\f_2 & iZ & i\dt\f_1 \\[0pt]
\j_4
 & iZ & -i\dt\f_2 & iF_1 \\[0pt]
\j_5
 & i\dt\f_4 & iF_3 & iF_2 \\[0pt]
\j_6
 & -i\dt\f_5 & iF_2 & -iF_3 \\[0pt]
\j_7
 & iF_3 & -i\dt\f_4 & -i\dt\f_5 \\[0pt]
\j_8
 & iF_2 & i\dt\f_5 & -i\dt\f_4 \\[0pt]
    \bottomrule
  \end{array}
  \label{e:GKYa}
\end{equation}
 which is clearly adinkraic: each instance of the $Q_I$-transformation of any one field is a monomial in terms of the component fields and their derivatives; this is depicted in Figure~\ref{f:GKYa}.
\begin{figure}[ht]
\centering
 \begin{picture}(140,45)
   \put(0,0){\includegraphics[width=140mm]{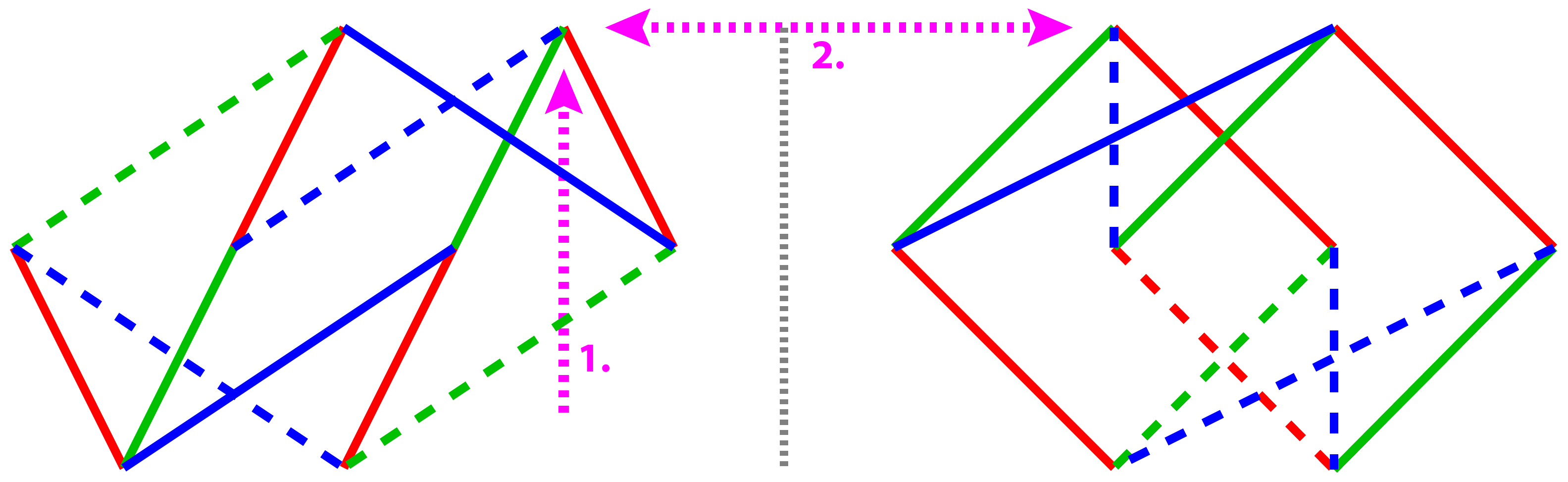}}
    \put(12,1){\cB{$\f_1$}}
    \put(31,1){\cB{$\f_2$}}
    \put(50.5,1){\cB{$\f_3$}}
    \put(99.5,1){\cB{$\f_4$}}
    \put(119,1){\cB{$\f_5$}}
    \put(2,20){\bB{$\j_1$}}
    \put(21,20){\bB{$\j_2$}}
    \put(40,20){\bB{$\j_3$}}
    \put(59,20){\bB{$\j_4$}}
    \put(80,20){\bB{$\j_5$}}
    \put(99,20){\bB{$\j_6$}}
    \put(118,20){\bB{$\j_7$}}
    \put(137,20){\bB{$\j_8$}}
    \put(31,40){\cB{$F_1$}}
    \put(99.5,40){\cB{$F_3$}}
    \put(119,40){\cB{$F_2$}}
    \put(50.5,40){\cB{$Z$}}
 \end{picture}
\caption{A graphical depiction of the adinkrized supermultiplet of Ref.\protect\cite{rTHGK12}. The two numbered arrows indicate the ({\bf1})~raising $\f_3$, and ({\bf2})~combining $\protect\dt\f_3{+}F_3$ into $Z$; it is the latter, now local field redefinition that disconnects the two Adinkras.}
 \label{f:GKYa}
\end{figure}
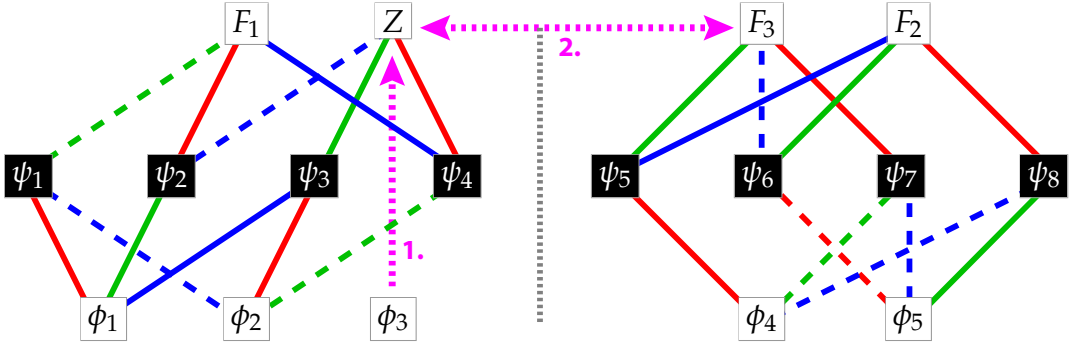
Note that the single component field redefinition\eq{e:f3r} procedure transformed:
\begin{enumerate}\itemsep=-3pt\vspace{-2mm}
 \item the $(5|8|3)$-component supermultiplet\eq{e:GKY} into a $(4|8|4)$-component one\eq{e:GKYa},
 \item which decomposes even before completing the procedure in the proof of Theorem~\ref{T:main}.
\end{enumerate}
By\eq{e:mindim}, the minimal representation of the $N\,{=}\,3$-extended worldline supersymmetry has four bosons and four fermions, and the two Adinkras in Figure~\ref{f:GKYa} are indeed minimal.

Reversing this transformation, we can start with the adinkraic supermultiplet\eq{e:GKYa}, and lower the linear combination
\begin{equation}
   \big(\ell\Defl Z-F_3\big) \mapsto \dt\f_3
   \qquad\textit{i.e.},\qquad
   \f_3 \Defl \int\!\rd\t\,\big(Z(\t)-F_3(\t)\big) + \textit{const.}
 \label{e:GKYnl}
\end{equation}
to achieve the gauge-quotient supermultiplet of Refs.\cite{rTHGK12,rTHGK13}. Notice that the integration constant recovers the constant term in $\f_3$ which was lost in raising $\f_3$\eq{e:f3r}. Thus, the non-adinkraic supermultiplet\eq{e:GKY} may indeed be understood as being a non-local component field transformation\eq{e:GKYnl} of the direct sum of two minimal Adinkras in Figure~\ref{f:GKYa}.

The reader might wonder whether perhaps the non-local transformation\eq{e:GKYnl} in fact somehow establishes an effective equivalence of\eq{e:GKY} and\eq{e:GKYa}. To show that this is not so, Ref.\cite{rTHGK13} constructs a 13-parameter family of Lagrangians (even while restricting to just bilinear terms!), where---for generic choices\ft{These choices do span a 13-dimensional open region where the Lagrangians define unitary models\cite{rTHGK13}.} of the 13 parameters---the following holds:
 \begin{quote}
   The generic (even if just bilinear) Lagrangians\cite{rTHGK13} for the supermultiplet\eq{e:GKY} depend on the component field $\f_3$ in ways that the transformation\eq{e:GKYnl} cannot be used to eliminate $\f_3 \to Z$ without rendering the generic Lagrangian non-local.
 \end{quote}
These ultimately dynamical considerations prove that the supermultiplet\eq{e:GKY} must be considered physically inequivalent from\eq{e:GKYa}. In turn, Theorem~\ref{T:main} provides a direct but non-local relationship suitable for classification purposes; see below.

\subsection{The Complex Linear Supermultiplet}
Ref.\cite{rGHHS-CLS} analyzes a well-documented representation of simple supersymmetry in 4-dimensional spacetime, the complex linear supermultiplet\cite{r1001}, dimensionally reduced to the $n\,{=}\,4$-extended supersymmetry of 1-dimensional worldline. This supermultiplet is defined by way of the complex superfield satisfying the quadratic superdifferential constraint $\bar{D}^{\dot\a}\bar{D}_{\dot\a}\S=0$, and Ref.\cite{rGHHS-CLS} then traces through the ensuing conditions and identifications imposed on the real and imaginary parts of the complex component fields, settling finally on a basis that maximally simplifies the result. Table~\ref{t:CLS} presents the supersymmetry transformations in a slightly adapted version of this optimal basis,
\begin{table}[htb]
\caption{The supersymmetry transformation rules in the complex linear supermultiplet dimensionally reduced to the 1-dimensional worldline, adapted from Ref.\protect\cite{rGHHS-CLS}}
\vspace*{-3mm}\footnotesize
$$\begin{array}{@{} c|cccc @{}}
 & \C1{Q_1} & \C2{Q_2} & \C3{Q_3}  & \C4{Q_4}\strut\\ 
    \toprule\vspace*{-2pt}
  K &  \z_1  &  \z_2 &  \z_3 &  \z_4\\
  L & - \rho_4 &  \rho_3 & - \rho_2 &  \rho_1\\
 \midrule
 X_{12} & \dt{\z}_2 &-\dt{\z}_1 & \dt{\z}_4\Gr{{-}\b_3} & -\dt{\z}_3\Gr{{-}\b_4}\\
 X_{14} & \dt{\z}_4 & \dt{\z}_3\Gr{{+}\b_3} & -\dt{\z}_2\Gr{{-}\b_1} & -\dt{\z}_1\\
 X_{24} &-\dt{\z}_3\Gr{{-}\b_4} & \dt{\z}_4 & \dt{\z}_1\Gr{{-}\b_2} & -\dt{\z}_2\\[1pt]
 Y_{12} & \dt{\r}_3 & \dt{\r}_4 & \dt{\r}_1\Gr{{-}\b_2} & \dt{\r}_2\Gr{{+}\b_1}\\
 Y_{14} & \dt{\r}_1 & -\dt{\r}_2\Gr{{-}\b_1} & -\dt{\r}_3\Gr{{-}\b_4} & \dt{\r}_4\\
 Y_{24} & \dt{\r}_2\Gr{{+}\b_1} & \dt{\r}_1 & -\dt{\r}_4\Gr{{+}\b_3} & -\dt{\rho}_3\\[1pt]
  \cline{2-5}\rule{0pt}{2.5ex}
 Z_1 & \b_1 & \b_2 & \b_3 & \b_4 \\
 Z_2 & \b_2 &-\b_1 &-\b_4 & \b_3 \\
 Z_3 & \b_3 & \b_4 &-\b_1 &-\b_2 \\
 Z_4 & \b_4 &-\b_3 & \b_2 &-\b_1 \\
    \bottomrule
 \end{array}\qquad
\begin{array}{@{} c|cccc @{}}
  & \C1{Q_1} & \C2{Q_2} & \C3{Q_3}   & \C4{Q_4}\strut\\ 
    \toprule
  \z_1   &  i\dt{ K} & -iX_{12} & i(X_{24}\Gr{{+}Z_4}) & -iX_{14}\\
  \z_2   &  iX_{12} & i\dt{ K}  & -i(X_{14}\Gr{{-}Z_3}) & -iX_{24}\\
  \z_3   & -i(X_{24}\Gr{{+}Z_4}) & i(X_{14}\Gr{{-}Z_3})  & i\dt{K}  & -i(X_{12}\Gr{{+}Z_1})\\
  \z_4   &  iX_{14} & iX_{24}  & i(X_{12}\Gr{{+}Z_1})  & i \dt{K}\\
  \r_1   &  iY_{14} & iY_{24} & i(Y_{12}\Gr{{+}Z_4}) & i\dt{L}\\
  \r_2   &  i(Y_{24}\Gr{{-}Z_1}) & -i(Y_{14}\Gr{{-}Z_2})  & -i\dt{L} & i(Y_{12}\Gr{{+}Z_4})\\
  \r_3   &  iY_{12} & i\dt{L}  & -i(Y_{14}\Gr{{-}Z_2})  & -iY_{24}\\
  \r_4   & -i\dt{L} & iY_{12}  & -i(Y_{24}\Gr{{-}Z_1}) & iY_{14}\\
 \midrule
 \b_1 & i\dt{Z}_1 & -i\dt{Z}_2 & -i\dt{Z}_3 & -i\dt{Z}_4 \\
 \b_2 & i\dt{Z}_2 & i\dt{Z}_1  & i\dt{Z}_4  & -i\dt{Z}_3 \\
 \b_3 & i\dt{Z}_3 & -i\dt{Z}_4 & i\dt{Z}_1  & i\dt{Z}_2 \\
 \b_4 & i\dt{Z}_4 & i\dt{Z}_3  & -i\dt{Z}_2 & i\dt{Z}_1 \\
    \bottomrule
  \end{array}$$
\label{t:CLS}
\end{table}
and which is also faithfully depicted by the graph in Figure~\ref{f:CLS}.
\begin{figure}[htb]
 \begin{center}
  \begin{picture}(160,55)(0,14)
   \put(0,7){\includegraphics[width=160mm]{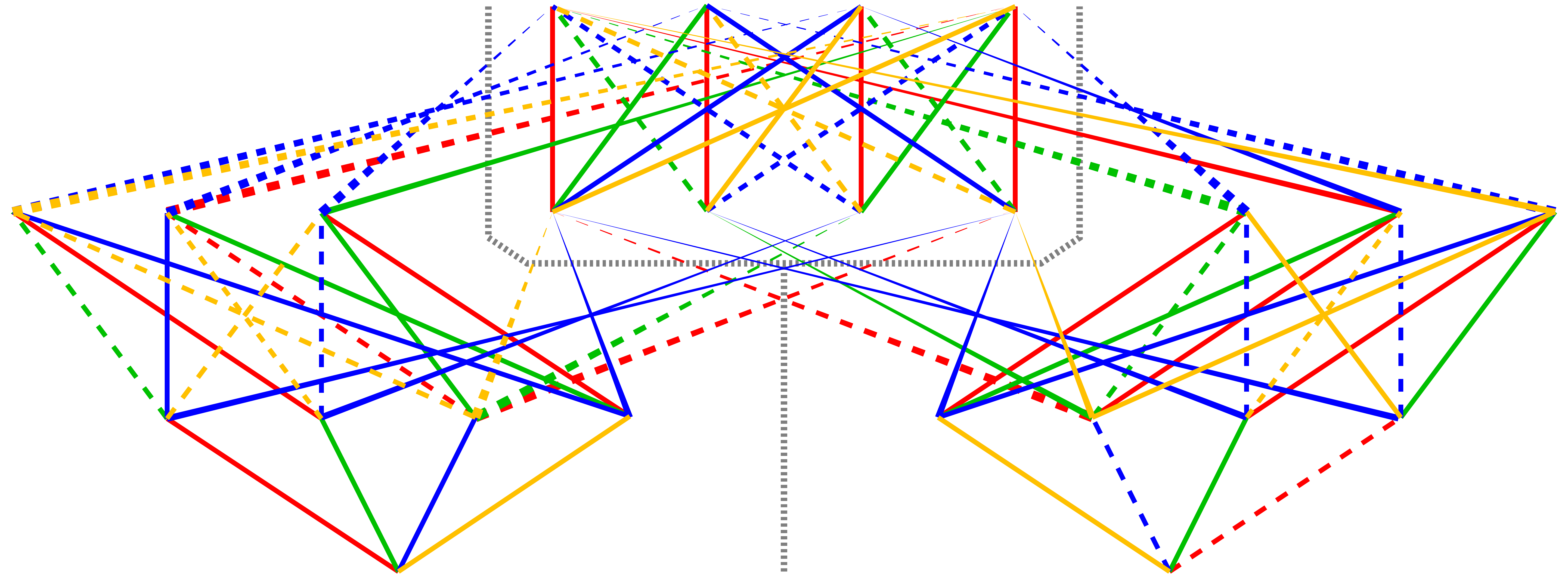}}
     \put(57,66){\bB{$\b_1$}}
     \put(72,66){\bB{$\b_2$}}
     \put(88,66){\bB{$\b_3$}}
     \put(103,66){\bB{$\b_4$}}
     \put(4,43){\cB{$X_{12}$}}
     \put(19,43){\cB{$X_{24}$}}
     \put(34,43){\cB{$X_{14}$}}
     \put(57,43){\cB{$Z_1$}}
     \put(72,43){\cB{$Z_2$}}
     \put(88,43){\cB{$Z_3$}}
     \put(103,43){\cB{$Z_4$}}
     \put(126,43){\cB{$Y_{14}$}}
     \put(141,43){\cB{$Y_{24}$}}
     \put(156,43){\cB{$Y_{12}$}}
     \put(17,22.5){\bB{$\z_1$}}
     \put(33,22.5){\bB{$\z_2$}}
     \put(49,22.5){\bB{$\z_3$}}
     \put(65,22.5){\bB{$\z_4$}}
     \put(96,23){\bB{$\r_1$}}
     \put(112,23){\bB{$\r_2$}}
     \put(127,23){\bB{$\r_3$}}
     \put(143,23){\bB{$\r_4$}}
     \put(40.5,7){\cB{$K$}}
     \put(119.5,7){\cB{$L$}}
  \end{picture}
 \end{center}
 \caption{A graph depicting the worldline ``shadow'' of the complex linear supermultiplet, adapted from Ref.\protect\cite{rGHHS-CLS}}
 \label{f:CLS}
\end{figure}
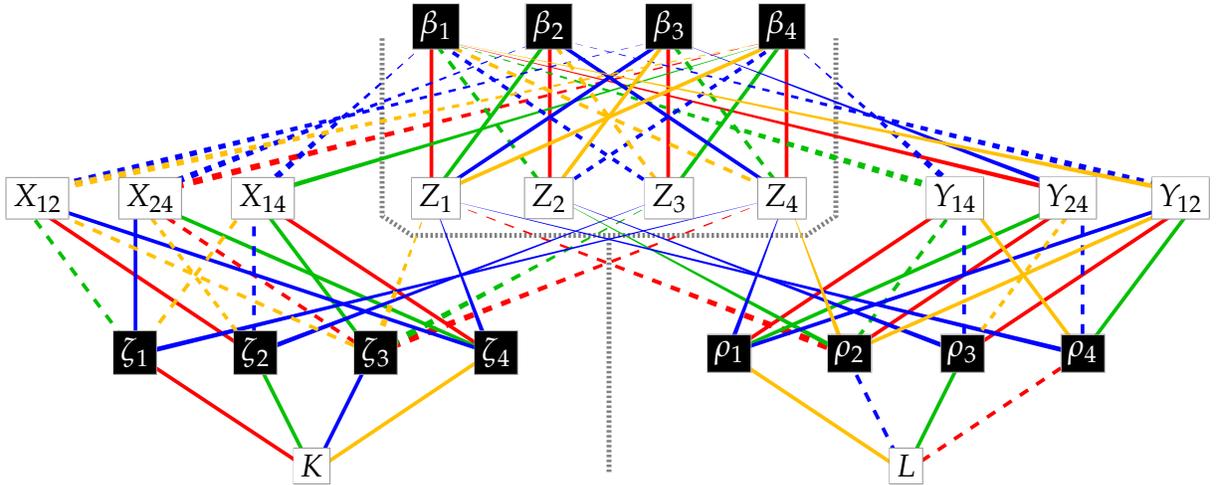

Notably, a quarter of the $Q_I$-transformations of individual component fields are binomials rather than monomials, and each of these binomials is depicted by two edges, one standard, the other (depicting a ``one-way'' action) tapered. Already the number (twenty-four) of binomials in Table~\ref{t:CLS} should indicate that the procedure of the proof of Theorem~\ref{T:main} will be considerably more involved than in the previous example. However, the statement of Theorem~\ref{T:main} also indicates a powerful tool in deciphering the optimal strategy---and even the eventual outcome.

Namely, Theorem~\ref{T:main} provides that, upon an adequate number of component field raising operations, the supermultiplet decomposes into a direct sum of {\em\/minimal\/} valise Adinkras. In turn, Ref.\cite{r6-3,r6-3.1} proved that in these minimal Adinkras, certain precisely specified higher order $Q_I$-operators act as quasi-projection operators. For the case at hand, in $N\,{=}\,4$ supersymmetry, these quasi-projection operators are
\begin{equation}
   \P^\pm \Defl \big[\,\C4{Q_4}\C3{Q_3}\C2{Q_2}\C1{Q_1}~\pm~\vdt^{~2}\,\big].
 \label{e:P}
\end{equation}
Closely related to the $\P^\pm$, Ref.\cite{r6-1.2} proves that the triple of operators
\begin{equation}
   \S^\pm_{IJ} \Defl \big[\,Q_IQ_J~\pm~\frc12\ve_{IJ}{}^{KL}Q_KQ_L\,\big]
   ~~\5{\sss\text{1--1}}{\longleftrightarrow}~~ \P^\pm
 \label{e:S}
\end{equation}
carries the same information about the structure (named ``chromotopology''\cite{r6-3}) of the $Q_I$-transformations. In particular, the two possible relative signs in the operators\eq{e:P} and\eq{e:S} are indeed indicative of the two equivalence classes of minimal Adinkras for $N\,{=}\,4$. In turn, these are well familiar from physics literature as being exemplified by the ``chiral'' and ``twisted-chiral'' superfields\cite{rTwSJG0,rGHR}, and we retain those names for the two classes.

 Now, the basis used for Table~\ref{t:CLS} is definitely not the one in which the decomposition is made manifest, even after we raise the component fields: $K,L$ into $\dt{K},\dt{L}$, and then $\z_i,\r_i$ into $\dt\z_i,\dt\r_i$. However, the operators\eq{e:P} and\eq{e:S} simplify the task of finding this basis as follows.
\begin{enumerate}\itemsep=-3pt\vspace{-2mm}
 \item Apply the quadratic operators\eq{e:S} on any component field. For example,
\begin{equation}
  [\C1{Q_1}\C2{Q_2}\pm\C3{Q_3}\C4{Q_4}]K = i\dt{K}\pm i\dt{K},
   \qquad\text{while}\qquad
  [\C1{Q_1}\C2{Q_2}\pm\C3{Q_3}\C4{Q_4}]Z_1 = i\dt{Z}_2\mp i\dt{Z}_2.
\end{equation}
For these to vanish, we must chose the lower sign for $K$ but the upper sign for $Z_1$, indicating that these two component fields will (upon some field redefinitions, perhaps) belong to distinct minimal Adinkras.

 \item Apply the quartic operators\eq{e:P} to any component field. For instance\ft{The operators\eq{e:P} and\eq{e:S} are defined to have the relative signs flipped if applied on fermions.},
\begin{subequations}
 \label{e:Redef}
\begin{alignat}9
  [\C4{Q_4}\C3{Q_3}\C2{Q_2}\C1{Q_1}\pm\vdt^{~2}]K
  &={\ttt\genfrac{\{}{\}}{0pt}{}{2}{0}}\ddt{K}{-}\dt{Z}_2, &\quad
  [\C4{Q_4}\C3{Q_3}\C2{Q_2}\C1{Q_1}\pm\vdt^{~2}]L
  &= \dt{Z}_3+{\ttt\genfrac{\{}{\}}{0pt}{}{2}{0}}\ddt{L}, \label{e:KL}\\
  [\C4{Q_4}\C3{Q_3}\C2{Q_2}\C1{Q_1}\pm\vdt^{~2}]X_{12}
  &=\dt{Z}_1+{\ttt\genfrac{\{}{\}}{0pt}{}{2}{0}}\ddt{X}_{12}, &\quad
  [\C4{Q_4}\C3{Q_3}\C2{Q_2}\C1{Q_1}\pm\vdt^{~2}]X_{14}
  &={\ttt\genfrac{\{}{\}}{0pt}{}{2}{0}}\ddt{X}_{14}{-}\dt{Z}_3, \label{e:XX}\\
 [\C4{Q_4}\C3{Q_3}\C2{Q_2}\C1{Q_1}\mp\vdt^{~2}]\z_1
  &= \dt\b_2-{\ttt\genfrac{\{}{\}}{0pt}{}{2}{0}}\ddt\z_1, &\quad
 [\C4{Q_4}\C3{Q_3}\C2{Q_2}\C1{Q_1}\mp\vdt^{~2}]\z_2
  &= -\dt\b_2-{\ttt\genfrac{\{}{\}}{0pt}{}{2}{0}}\ddt\z_2,  \label{e:zi}\\
  \text{but}\quad
  [\C4{Q_4}\C3{Q_3}\C2{Q_2}\C1{Q_1}\pm\vdt^{~2}]Z_i
  &\propto(\ddt{Z}_i\mp\ddt{Z}_i), &\quad
  [\C4{Q_4}\C3{Q_3}\C2{Q_2}\C1{Q_1}\mp\vdt^{~2}]\b_i
  &=\propto(\ddt\b_i\mp\ddt\b_i). \label{e:Zb}
\end{alignat}
\end{subequations}
The fact that the quasi-projections\eq{e:Zb} on $(Z_i|\b_i)$ act as $\vdt$-multiples of the identity for the lower choice of the sign indicate that these fields already form a separate $(4|4)$-component minimal valise Adinkra. This is visible from the graph in Figure~\ref{f:CLS}, upon noticing that no edges emanate from this (top-most) portion, as was already noted in Ref.\cite{rGHHS-CLS}.

Next, given now that $(Z_i|\b_i)$ form a separate sub-supermultiplet, the relations such as\eq{e:KL}, (\ref{e:XX}) and\eq{e:zi} indicate that:
 ({\bf1})~the right-hand side field combinations should be used in component field redefinitions, and
 ({\bf2})~we {\em\/must\/} use the upper sign-choices, not to duplicate the fields $(Z_i|\b_i)$.
\end{enumerate}

Following the procedure in the proof of Theorem~\ref{T:main}, we first notice that the bosonic component fields $K,L$ satisfy the conditions for the ``raising'' part of Definition~\ref{D:R/L}, and we raise them. In the resulting supermultiplet, now the fermionic components $\z_1,{\cdots},\z_4,\rho_1,{\cdots}\rho_4$ all satisfy the conditions of the ``raising'' part of Definition~\ref{D:R/L}, and we raise them too. This results in a valise supermultiplet with 12 bosons, all with the same engineering dimension, and 12 fermions, all with the engineering dimension $\frc12$ higher than the bosons. According to Theorem~\ref{T:valise}, this decomposes into a direct sum of minimal valise Adinkras. We exhibit this decomposition by employing the component field redefinitions given by the computations of the type\eq{e:Redef}:
\begin{subequations}
 \label{e:CLS>A}\small
\begin{alignat}9
  Z_1 &\Defl  (\2X_{34}{-}X_{12}),&\quad
  Z_2 &\Defl  (\2Y_{14}{-}Y_{23}),&\quad
  Z_3 &\Defl  (\2X_{14}{-}X_{23}),&\quad
  Z_4 &\Defl  (\2Y_{34}{-}Y_{12});\\
  X_1 &\Defl  (2\dt{\2K}{-}Y_{14}{+}Y_{23}),&\quad
  X_2 &\Defl  (\2X_{12}{+}X_{34}),&\quad
  X_3 &\Defl -(2\2X_{24}{+}Y_{34}{-}Y_{12}),&\quad
  X_4 &\Defl  (\2X_{23}{+}X_{14});\\
  Y_1 &\Defl  (2\dt{\2L}{+}X_{14}{-}X_{23}),&\quad
  Y_2 &\Defl  (\2Y_{12}{+}Y_{34}),&\quad
  Y_3 &\Defl -(2\2Y_{24}{-}X_{34}{+}X_{12}),&\quad
  Y_4 &\Defl  (\2Y_{23}{+}Y_{14});\\
 \x_1 &\Defl  2\dt{\2\z}{}_1{-}\b_2,&\qquad
 \x_2 &\Defl  2\dt{\2\z}{}_2{+}\b_1,&\qquad
 \x_3 &\Defl  2\dt{\2\z}{}_3{+}\b_4,&\qquad
 \x_4 &\Defl  2\dt{\2\z}{}_4{-}\b_3;\\
 \h_1 &\Defl -2\dt{\2\r}{}_4{+}\b_3,&\qquad
 \h_2 &\Defl  2\dt{\2\r}{}_3{+}\b_4,&\qquad
 \h_3 &\Defl -2\dt{\2\r}{}_2{-}\b_1,&\qquad
 \h_4 &\Defl  2\dt{\2\r}{}_1{-}\b_2,
\end{alignat}
\end{subequations}
where we have underlined the component fields that are being eliminated by each definition. The $Q_I$-transformations in this (\ref{e:CLS>A})-transform of the complex linear supermultiplet decouple:
\begin{equation}
\begin{array}{@{} c|cccc @{}}
 \omit& \C1{Q_1}  & \C2{Q_2} & \C3{Q_3} & \C4{Q_4} \strut\\ 
    \toprule
 X_1 &  \x_1 &  \x_2 &  \x_3  &  \x_4 \\[1pt]
 X_2 &  \x_2 & -\x_1 &  \x_4  & -\x_3 \\[1pt]
 X_3 &  \x_3 & -\x_4 & -\x_1  &  \x_2 \\[1pt]
 X_4 &  \x_4 &  \x_3 & -\x_2  & -\x_1 \\[1pt]
  \hline
 Y_1 &  \h_1 &  \h_2 &  \h_3 &  \h_4 \\[1pt]
 Y_2 &  \h_2 & -\h_1 &  \h_4 & -\h_3 \\[1pt]
 Y_3 &  \h_3 & -\h_4 & -\h_1 &  \h_2 \\[1pt]
 Y_4 &  \h_4 &  \h_3 & -\h_2 & -\h_1 \\[1pt]
  \hline
 Z_1    &  \b_1 &  \b_2  & \b_3  &  \b_4 \\
 Z_2    &  \b_2 & -\b_1  & -\b_4 &  \b_3 \\
 Z_3    &  \b_3 &  \b_4  & -\b_1 & -\b_2 \\
 Z_4    &  \b_4 & -\b_3  & \b_2  & -\b_1 \\
    \bottomrule
  \end{array}
  \qquad
\begin{array}{@{} c|cccc @{}}
 \omit& \C1{Q_1}  & \C2{Q_2} & \C3{Q_3} & \C4{Q_4} \strut\\ 
    \toprule
 \x_1   &  i\dt{X}_1 & -i\dt{X}_2  & -i\dt{X}_3 & -i\dt{X}_4\\
 \x_2   &  i\dt{X}_2 &  i\dt{X}_1  & -i\dt{X}_4 &  i\dt{X}_3\\
 \x_3   &  i\dt{X}_3 &  i\dt{X}_4  &  i\dt{X}_1 & -i\dt{X}_2\\
 \x_4   &  i\dt{X}_4 & -i\dt{X}_3  &  i\dt{X}_2 &  i\dt{X}_1\\
  \hline\rule{0pt}{2.5ex}
 \h_1   &  i\dt{Y}_1 & -i\dt{Y}_2  & -i\dt{Y}_3  & -i\dt{Y}_4\\
 \h_2   &  i\dt{Y}_2 &  i\dt{Y}_1  & -i\dt{Y}_4  &  i\dt{Y}_3\\
 \h_3   &  i\dt{Y}_3 &  i\dt{Y}_4  &  i\dt{Y}_1  & -i\dt{Y}_2\\
 \h_4   &  i\dt{Y}_4 & -i\dt{Y}_3  &  i\dt{Y}_2  &  i\dt{Y}_1\\
  \hline\rule{0pt}{2.5ex}
  \b_1  &  i\dt{Z}_1  & -i\dt{Z}_2  & -i\dt{Z}_3  & -i\dt{Z}_4 \\
  \b_2  &  i\dt{Z}_2  &  i\dt{Z}_1  &  i\dt{Z}_4  & -i\dt{Z}_3 \\
  \b_3  &  i\dt{Z}_3  & -i\dt{Z}_4  &  i\dt{Z}_1  &  i\dt{Z}_2 \\
  \b_4  &  i\dt{Z}_4  &  i\dt{Z}_3  & -i\dt{Z}_2  &  i\dt{Z}_1\\
    \bottomrule
  \end{array}
 \label{e:CLS3}
\end{equation}
and may be depicted by the three minimal $N\,{=}\,4$ Adinkras given in Figure~\ref{f:CLS3}. The middle one is ``twisted'' as compared to the flanking ones: in it, all edges of the fourth color have their solidness/dashedness flipped, which corresponds to swapping $\C4{Q_4}\to-\C4{Q_4}$. Correspondingly, the middle Adinkra in Figure~\ref{f:CLS3} is annihilated by the operators\eq{e:P} and\eq{e:S} with one choice of the relative sign, while the flanking Adinkras are annihilated by these operators with the opposite choice of the relative sign. This agrees with the computations of Ref.\cite{rUMD09-2}.
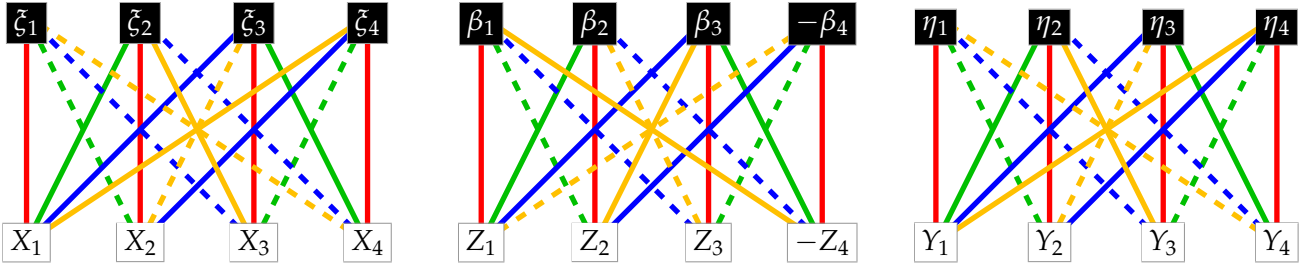
\begin{figure}[ht]
\centering
 \begin{picture}(173,35)\small\unitlength=.944444mm
   \put(0,0){\includegraphics[width=170mm]{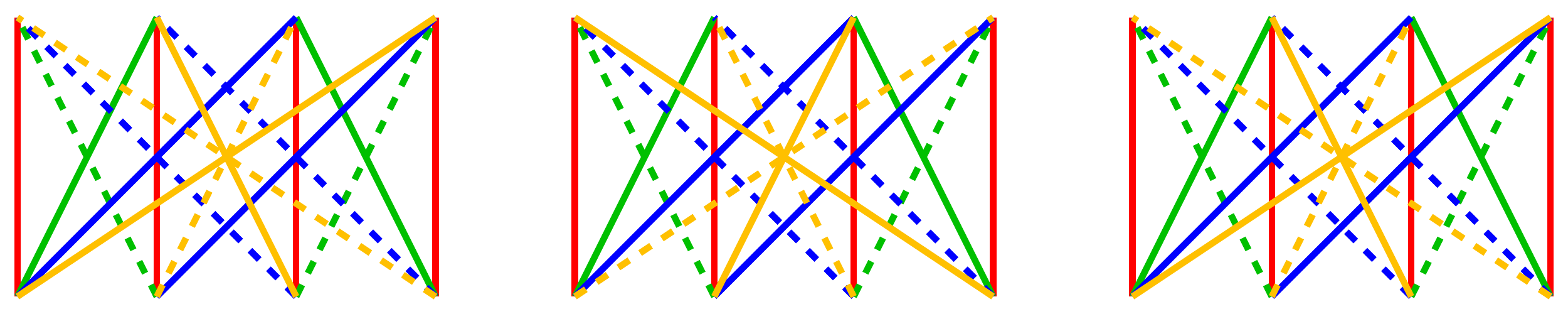}}
    \put(2,1){\cB{$X_1$}}
    \put(18,1){\cB{$X_2$}}
    \put(34,1){\cB{$X_3$}}
    \put(50,1){\cB{$X_4$}}
    \put(2,32){\bB{$\x_1$}}
    \put(18,32){\bB{$\x_2$}}
    \put(34,32){\bB{$\x_3$}}
    \put(50,32){\bB{$\x_4$}}
    \put(66,1){\cB{$Z_1$}}
    \put(82,1){\cB{$Z_2$}}
    \put(98,1){\cB{$Z_3$}}
    \put(114,1){\cB{$-Z_4$}}
    \put(66,32){\bB{$\b_1$}}
    \put(82,32){\bB{$\b_2$}}
    \put(98,32){\bB{$\b_3$}}
    \put(114,32){\bB{$-\b_4$}}
    \put(130,1){\cB{$Y_1$}}
    \put(146,1){\cB{$Y_2$}}
    \put(162,1){\cB{$Y_3$}}
    \put(178,1){\cB{$Y_4$}}
    \put(130,32){\bB{$\h_1$}}
    \put(146,32){\bB{$\h_2$}}
    \put(162,32){\bB{$\h_3$}}
    \put(178,32){\bB{$\h_4$}}
 \end{picture}
\caption{The three separate minimal Adinkras into which the worldline ``shadow'' of the complex linear supermultiplet may be transformed by means of iterative raising and linear combinations.}
 \label{f:CLS3}
\end{figure}

Finally, the original, worldline ``shadow'' of the complex linear supermultiplet (Table~\ref{t:CLS} and Figure~\ref{f:CLS}) is then reconstructed by applying the (non-local) inverse of the component field redefinitions\eq{e:CLS>A} on the direct sum Adinkra\eq{e:CLS3}, depicted in Figure~\ref{f:CLS3}.
 Note that the component field redefinitions\eq{e:CLS>A} lose the constant terms in the original component fields $K,L,\z_i,\r_i$, and that the inverse of\eq{e:CLS>A} then re-supplies these constant terms by way of integration constants. For example,
\begin{equation}
  K(\t) = \inv2\int\!\rd\t\,\big[X_1(\t)+Z_2(\t)\big]+K_0,
   \qquad
  \z_2(\t) = \inv2\int\!\rd\t\,\big[\x_2(\t)-\b_1(\t)\big]+\z_{2,0},
\end{equation}
and so on.

\section{A Continuum of Supermultiplets, and Their Classification}
The fact that the procedure of the proof of Theorem~\ref{T:valise} employs linear combinations of component fields where the coefficients are not restricted to be integers is rather suggestive.
 Indeed, we now demonstrate how a continuum of $N=5$ supermultiplets may be obtained from a certain $N=5$ Adinkra, by lowering a non-trivial linear combination of nodes.

We begin with the $N=5$ Adinkraic supermultiplet with the following field content:
\begin{equation}
  (A\,|\,\psi_I\,|\,B_{IJ}, V_I\,|\,\c_{IJ},\w)
  \qquad I,J=1,{\cdots}\,5,
\end{equation}
where the indices are $SO(5)$ indices from the supersymmetry, $B_{IJ}=-B_{JI}$ and $\c_{IJ}=-\c_{JI}$.  The supersymmetry transformations within this supermultiplet are given as follows:
\begin{subequations}
 \label{e:N5}
\begin{align}
 Q_I A      &= \j_I, \\
 Q_J \j_I   &= i\, B_{IJ}+i\,\d_{IJ}\dt{A}, \\
 Q_K B_{IJ} &= \frc12\ve_{IJK}{}^{LM}\c_{LM} + 2\d_{K[J}\dt\j_{I]}, \\
 Q_K \c_{IJ}&= 2i\,\d_{K[J}\dt{V}_{I]} + \frc{i}2\,\ve_{IJK}{}^{LM} \dt{B}_{LM}, \\
 Q_J V_I    &= \d_{IJ}\w + \c_{IJ}, \\
 Q_I \w     &=i\,\dt{V}_I,
\end{align}
\end{subequations}
where bracketed indices are antisymmetrized with weight $\frc12$: $\d_{K[J}\dt\j_{I]}\Defl\frc12(\d_{KJ}\dt\j_{I}{-}\d_{KI}\dt\j_{J})$.
As written here, it may not be immediately obvious that the result of the application of any one supercharge on any one component field in fact is a monomial. We thus include a sampling of these supersymmetry transformation rules below to illustrate this fact, \ie, that this supermultiplet is Adinkraic:
\begin{equation}
\begin{array}{c|ccccc}
 & \C1{\bm{Q_1}} &\C2{\bm{Q_2}} &\C3{\bm{Q_3}} &\C4{\bm{Q_4}} &\C5{\bm{Q_5}}\\
 \toprule
 A & \j_1 & \j_2&\j_3&\j_4&\j_5\\[2pt]
 \hline\rule{0pt}{2.5ex}
B_{12} & -\dt{\j}_2 & \dt{\j}_1 & \c_{45} & -\c_{35} & \c_{34}\\[1pt]
B_{13} & -\dt{\j}_3 & -\c_{45} & \dt{\j}_1 & \c_{25} & -\c_{24}\\
 \vdots\\
 \hline\rule{0pt}{2.5ex}
V_1 &  \w      &  \c_{12} &  \c_{13} &  \c_{14} & \c_{15}\\[1pt]
V_2 & -\c_{12} &  \w      &  \c_{23} &  \c_{24} & \c_{25}\\
\vdots\\
 \bottomrule
\end{array}
 \qquad
\begin{array}{c|ccccc}
 & \C1{\bm{Q_1}} &\C2{\bm{Q_2}} &\C3{\bm{Q_3}} &\C4{\bm{Q_4}} &\C5{\bm{Q_5}}\\
 \toprule
\j_1 & i\dt{A} & iB_{12} & iB_{13} & iB_{14} & iB_{15}\\
\j_2 & -iB_{12} & i\dt{A} & iB_{23} & iB_{24} & iB_{25}\\
 \vdots\\
 \hline\rule{0pt}{2.5ex}
\c_{12} & -i\dt{V}_2 & i\dt{V}_1 & i\dt{B}_{45} & -i\dt{B}_{35} &  i\dt{B}_{34}\\
\c_{13} & -i\dt{V}_3 & -i\dt{B}_{45} & i\dt{V}_1 & i\dt{B}_{25} & -i\dt{B}_{24}\\
 \vdots\\
 \hline\rule{0pt}{2.5ex}
\w & i\dt{V}_1 & i\dt{V}_2 & i\dt{V}_3 & i\dt{V}_4 & i\dt{V}_5\\
 \bottomrule
\end{array}
 \label{e:N5Q}
\end{equation}
The interested Reader may complete the table using the above formulae for the transformation rules.
 Instead of drawing the Adinkra in full detail\cite{r6-1}, it is more illustrative this time to see a somewhat more collapsed diagram:
\begin{equation}\unitlength=.4mm\thicklines
 \vC{\begin{picture}(150,100)(0,-5)
       \put(10,10){\circle{5}}
       \put(0,7){$1$}
	   \put(11,0){$A$}
	   \put(12,12){\line(1,1){21}}
	   \put(35,35){\circle*{5}}
	   \put(25,35){$5$}
	   \put(39,25){$\j_I$}
	   \put(37,37){\line(1,1){21}}
	   \put(60,60){\circle{5}}
	   \put(48,63){$10$}
	   \put(64,50){$B_{IJ}$}
	   \put(62,62){\line(1,1){21}}
	   \put(85,85){\circle*{5}}
	   \put(73,88){$10$}
	   \put(78,68){$\c_{IJ}$}
	   \put(87,83){\line(1,-1){21}}
	   \put(110,60){\circle{5}}
	   \put(108,68){$5$}
	   \put(105,48){$V_I$}
	   \put(112,62){\line(1,1){21}}
	   \put(135,85){\circle*{5}}
	   \put(138,86){$1$}
	   \put(133,74){$\w$}
     \end{picture}}
 \label{e:N5A}
\end{equation}
where we have partially collapsed nodes as in Ref.\cite{rA} for convenience, but not all the way, to exhibit the different types of component fields more clearly.  Note that this can be obtained from the unconstrained real $N=5$ superfield by lowering the top fermion $\w$ and then lowering all components of the five-vector $V_I$.

We can now lower a linear combination of the $\c_{IJ}$ and $\w$ by defining for example:
\begin{equation}
  \skew{-5}\dt\l \Defl \cos(\q)\,\w +  \sin(\q)\,\c_{12}
 \label{e:cwq}
\end{equation}
and replace $\w$ with $\l$. That is, we consider the $(1|6|15|9)$-dimensional supermultiplet with the component field content
\begin{equation}
  \cM_\q \Defl (A\,|\,\psi_I, \lambda \,|\,B_{IJ}, V_I\,|\,\c_{IJ})
 \label{e:q}
\end{equation}
wherein the supersymmetry transformation rules are as in\eq{e:N5} except that:
\begin{enumerate}\itemsep=-3pt\vspace{-2mm}
 \item All appearances of $\w$ in\eq{e:N5}---which occur only in $Q_I(V_J)$ for $I=J$---are now replaced by
\begin{equation}
  \w~~\mapsto~~\sec(\q)\,\skew{-5}\dt\l - \tan(\q)\,\c_{12}.
 \label{e:K}
\end{equation}

 \item The row for $Q_I(\w)$ in\eq{e:N5} is now replaced by a row for $Q_I(\l)$, which results in:
\begin{equation}
  \begin{array}{c@{(\l)=\,\,}r@{\,}l}
   \C1{\bm{Q_1}} & i\cos(\q)\,V_1 & -\,i\sin(\q)\,V_2, \\
   \C2{\bm{Q_2}} & i\cos(\q)\,V_2 & +\,i\sin(\q)\,V_1, \\
   \C3{\bm{Q_3}} & i\cos(\q)\,V_3 & +\,i\sin(\q)\,B_{45}, \\
   \C4{\bm{Q_4}} & i\cos(\q)\,V_4 & -\,i\sin(\q)\,B_{35}, \\
   \C5{\bm{Q_5}} & i\cos(\q)\,V_5 & +\,i\sin(\q)\,B_{34}. \\
  \end{array}
 \label{e:Ql}
\end{equation}
\end{enumerate}
Thus, some of the supersymmetry transformation rules now involve linear combinations of fields. We may think of the $\sin\q$-terms in the equations\eq{e:cwq} and\eq{e:Ql} as generating a deformation of the adinkraic supermultiplet\eq{e:N5Q}, being ``tuned'' by the continuous angle $\q$.

To prove that there exists no local change of basis that can remove the occurrence of linear combinations from the results of applying the $Q_I$ on the component fields, we proceed by contradiction, assuming that there exists a basis in which the supermultiplet $\cM_\q$ is adinkraic, and systematically search for such a basis.

Start with $A$, and note that there is no other component field or $\t$-derivative thereof with the engineering dimension $[A]$, simply because $A$ has the lowest engineering dimension within the supermultiplet $\cM_\q$. The only possible field redefinition of $A$ is then a real re-scaling, $A\mapsto c_AA$, with some non-zero $c_A\in\IR$.
 Within an Adinkraic basis for $\cM_\q$, it would have to be that $Q_I(A)$ are all basis elements as well, so that the $\j_I$ are all (up to a multiplicative constant) also basis elements.  Proceeding in this way, the $B_{IJ}$, $\c_{IJ}$ and $V_I$ are likewise all basis elements, each one up to a multiplicative constant.  But now the $Q_I(V_J)$ involve a linear combination of $\skew{-3}\dt\l$ and $\c_{12}$:
\begin{equation}
   Q_I(V_J) = \d_{IJ}\,\big(\,\sec(\q)\,\skew{-3}\dt\l - \tan(\q)\,\c_{12}\,\big),
\end{equation}
failing the defining property of adinkraic supermultiplets unless $\q$ is an integral multiple of $\frc\p2$. To solve this, we would have to implement the inverse of the component field redefinition\eq{e:K},
\begin{equation}
  \l = \int\!\rd\t~\big( \cos(\q)\,\w + \sin(\q)\,\c_{12} \big) + \l_0
 \label{e:Kb}
\end{equation}
which is non-local, and so not allowed in general.

We thus conclude that the supermultiplet $\cM_\q$ is truly non-adinkrizable when $\q$ is not an integeral multiple of $\p$.\QED

Furthermore, this argument also shows that these supermultiplets are also inequivalent for different values of $\theta$.  This provides an example of a continuum of distinct supermultiplets.

 In principle, there may well exist usable Lagrangians for the supermultiplet $\cM_\q$, wherein the component field $\l$ only occurs with a $\t$-derivative acting on it, so that the non-locality of the transformation\eq{e:Kb} does not show up in the dynamics of these models. However, recent explicit computations for similar supermultiplets show that this is not the case in most general (still unitary) Lagrangians\cite{rTHGK13a}, and that observable couplings to external (probing) fields exist that would detect such non-locality\cite{rTHGK13}; the variety of employed value(s) of $\q$ is thus observable.
\ping

Reconsider now the example\eqs{e:N5}{e:N5A}. Instead of\eq{e:cwq}, we could lower any other, more general linear combination
\begin{equation}
  \sum_{IJ} c^{IJ} \c_{IJ} + c^0\w.
 \label{e:LinCo}
\end{equation}
Each particular of these continuously many choices would result in a distinct supermultiplet, thus defining an 11-parameter continuous family of $(1|6|15|9)$-dimensional supermultiplets generalizing\eq{e:q}. Modulo the overall scaling, these parameters form the projective space $\IR\IP^{10}$.

 Certain of these resulting supermultiplets may be shown to be equivalent to each other by local component field redefinitions alone. Further equivalence relations may be provided by allowing the supersymmetry charges $Q_I$ to rotate using their $SO(5)$ $R$-symmetry. Tempting as this latter option may be, note that one can easily construct models employing two or more distinct so-constructed supermultiplets, each with a different linear combination of $\c_{IJ}$ and $\w$. Using $R$-symmetries then cannot, in general, reduce the number of effective linear combination coefficients in {\em\/all\/} of the so-constructed supermultiplets to the same smaller number of effective coefficients in each supermultiplet simultaneously. Nor can one hope to be able to transform two distinct so-constructed supermultiplets {\em\/simultaneously\/} into any one particular version, even by using $R$-symmetries together with local component field redefinitions.
 
Ultimately, the determination whether or not $R$-symmetries provide admissible equivalence relations is then a fundamentally dynamical question, depending on the choice of a Lagrangian to govern that dynamics. Similarly, as discussed in Section~\ref{s:GKY}, whether or not even some non-local field redefinitions provide admissible equivalence relations depends on limitations (such as gauge symmetries) that one may impose on the choice of the Lagrangian. 
 
 In principle, the parameters in\eq{e:LinCo} provide a ``rough parameter space,'' wherein one is yet to identify the various points that correspond to equivalent supermultiplets. The transformations relating such points are expected to form a group usually called the ``mapping class group.''
 The above considerations however indicate that a proper definition of a such a mapping class group, which would reduce a rough parameter space of such linear combinations into a proper moduli space, is rather delicate a problem and well beyond the scope of this article. 
 Suffice it to say that finding this ``moduli space'' of supermultiplets would be to consider the projective space of lines in the space spanned by $\c_{IJ}$ and $\w$, and then quotienting by the action of $SO(5)$ should $R$-symmetries be physically/dynamically permissible.

Of course, it is also possible to lower several linearly independent combinations of the form\eq{e:LinCo}, rather than just one. Furthermore and depending on the result at this stage as pertaining to Definition~\ref{D:R/L}.b, one may next entertain the lowering of one or more nontrivial linear combinations
\begin{equation}
  \sum_{IJ} d^{IJ} B_{IJ} + \sum_I d^I V_I.
\end{equation}
and so on.

Quite generally then, to classify off-shell worldline supermultiplets, one could start with direct sums of Adinkras, then progressively lower linear combinations of nodes in all possible ways and any number of times, subject to the conditions of Definition~\ref{D:R/L}.b.  This leads to considering projective spaces if we lower a single linear combination, or Grassmannian manifolds if we lower many.  If we continue to lower linear combinations of component fields in this manner, there may emerge complicated interrelations between these lowered subspaces, resulting in certain specific types of flag-varieties.

 While the general classification of these constructions remains difficult, it is a straightforward project to generate all possible lowerings from any one particular direct sum of Adinkras.

\section{Conclusions}
This paper illustrates a method for writing any off-shell engineerable supermultiplet in one dimension as a finite number of lowerings from an Adinkra.  We have illustrated this method with two examples from recent literature, as well as an example specially constructed for this purpose.

Off-shell supermultiplets of $N\,{=}\,1$ and $N\,{=}\,2$ supersymmetry are in fact adinkrizable without the need for lowerings.  This can be proved using the following method: since we know that such supermultiplets can be obtained by lowering an Adinkra, examine Adinkras where a lowering of a linear combination of nodes is possible, then show that there exists a change of basis that results in the linear combination being a single node.

The situation with $N\,{=}\,3$ supermultiplets is somewhat special in a different way:  most of $N\,{=}\,3$ supermultiplets in fact admit a fourth supersymmetry, and sometimes even in several distinct ways\cite{rTHGK13}.  In fact, this is always possible except for the situation where an irreducible $N\,{=}\,3$ supermultiplet takes up four different engineering dimensions, in which case this is the real unconstrained superfield with $N\,{=}\,3$, and such a supermultiplet is not only adinkrizable, but in fact equal to the ``top Adinkra'' in the language of \cite{rA}.

With $N\,{=}\,4$ and higher, however, lowering linear combinations may be necessary.  This leads to an approach to classifying off-shell engineerable supermultiplets: start with a direct sum of minimal Adinkras, then map out all the ways of lowering linear combinations of nodes.  Equivalently, we can say that we are iteratively lowering linear subspaces spanned by these linear combinations. A classification of off-shell supermultiplets then must include a classification of possible choices of subspaces that can be iteratively lowered; we defer the exploration of this avenue to a subsequent effort.

\section{Acknowledgments}
We should like to thank M.~Faux and S.J.~Gates, Jr.\ for extended discussions on issues of supersymmetry, during which many of the concepts used in this paper were developed.
 CD acknowledges the support from the Natural Sciences and Engineering Resource Council of Canada, the Pacific Institute for the Mathematical Sciences, and the McCalla Professorship at the University of Alberta.
 TH thanks the Department of Physics, University of Central Florida, Orlando FL, and the Physics Department of the Faculty of Natural Sciences of the University of Novi Sad, Serbia, for the recurring hospitality and resources.
 GL acknowledges the support by a grant from the Simons Foundation (Award Number 245784).

\providecommand{\href}[2]{#2}
\begingroup\raggedright
\endgroup
\end{document}